\documentclass[11pt,envcountsame,envcountsect]{llncs}
 
\usepackage{graphicx}
\usepackage{amssymb}
\usepackage{epstopdf}
\usepackage{amsmath}
\usepackage{qip}
\usepackage{framed}
\usepackage{epsfig}
\usepackage{times}
\usepackage[a4paper,hmargin=1in,vmargin=1.7in]{geometry}

\newcommand{\sot}{$1$-$2$-OT}

\newcommand{\ot}{{\sc 1-2-ot}}

\newcommand{\dep}[2]{\ensuremath{#1 \searrow #2}}

\newcommand{\assign}{\ensuremath{\kern.5ex\raisebox{.1ex}{\mbox{\rm:}}\kern -.3em =}}
\newcommand{\hil}{\mathcal{H}}
\newcommand{\ID}[1]{\ensuremath{{\sf ID}_{#1}}}

\newcounter{itm}
\newenvironment{myprotocol}[1]
  {\begin{minipage}{\columnwidth}
    \begin{framed}\hspace{0ex}
     \begin{minipage}{0.99\columnwidth}
       {\bf #1:}
       \setcounter{itm}{1}
       \begin{list}{\arabic{itm}.}{\usecounter{itm}
          \setlength{\itemsep}{0mm}
          \setlength{\leftmargin}{\labelwidth}
          \setlength{\topsep}{\parsep}}}
   {    \end{list}
       \vspace{-1.5ex}
       \end{minipage}
     \end{framed}
    \end{minipage}\vspace{-0.6ex}}

\newenvironment{myfigure}[1]
         {\begin{figure}[#1] \centering}
         { \end{figure}}

\newcommand{\myparagraph}[1]{\medskip\noindent{\sc #1}}

\title{Two-Party Quantum Protocols Do Not Compose Securely %Even 
Against
 Honest-But-Curious Adversaries}
\author{
Louis Salvail\inst{1}\fnmsep\thanks{QUSEP, Quantum Security in Practice, funded by the
                                   Danish Natural Science Research Council.}
 \and Miroslava Sot\'{a}kov\'{a}\inst{2,3}
}

\institute{
Department of Computer Science, Universit\'e de Montr\'eal, QC, Canada\,
(\email{salvail@iro.umontreal.ca})
\and
Department of Computer Science, Aarhus University, Denmark\,
(\email{mirka@cs.au.dk})
\and
Research Center for Quantum Information, Slovak Academy of Sciences, Bratislava, Slovakia
}

\begin{document}
\maketitle

\begin{abstract}
In this paper, we build upon the model of two-party quantum computation introduced by Salvail 
et al.~\cite{SSS09} 
and show that in this model, only trivial correct two-party quantum protocols are weakly self-composable. We do so 
by defining a protocol $\Pi$ calling any non-trivial sub-protocol $\pi$ $N$ times 
and showing that there is a quantum honest-but-curious strategy that cannot be modeled 
by acting locally in every single copy of $\pi$. In order to achieve this, we assign 
a real value called \emph{payoff} to any strategy for $\Pi$ and show that that 
there is a gap between the highest payoff achievable by coherent and local strategies. 
\end{abstract}

\section{Introduction}
The most striking result in quantum cryptography is certainly
the capacity to perform secret-key distribution~\cite{BB84}
securely by a universally composable
quantum protocol~\cite{RK05,BHLMO05}.
This is in sharp contrast with what is achievable using classical
communication alone. A different class of cryptographic primitives,
called two-party computation, is not as easy to solve using quantum
communication. In fact, some two-party primitives are
as impossible to achieve using quantum communication as
they are based solely on clasical communication.
In particular,
%concerning the realization, 
well-known two-party primitives like oblivious
transfer~\cite{Lo97},
bit commitment~\cite{Mayers97,LC97}, and fair coin-tossing~\cite{Kitaev03}
have neither classical nor quantum secure implementations. However, there
exists weaker primitives achievable by quantum protocols
but impossible in the  classical world.
For instance, sharing an EPR pair allows for two players to
implement a noisy version of a two-party primitive called
\emph{non-local box} (NLB) \footnote{${\rm NLB}:(x^A,y^B) 
\mapsto (a^A,(a\oplus xy)^B)$, where $x$ and $y$ are Alice's and Bob's 
respective input bits, $a$ is a uniformly random output bit for Alice, 
and $a\oplus xy$ is the output bit for Bob.} with noise rate
$\sin^2\frac{\pi}{8}$~\cite{PR94,BLMPPR05}, which is a task impossible
to achieve classically. Due to the local equivalence
between randomized
NLB  and randomized \emph{one-out-of-two oblivious transfer} (\sot)
\footnote{\sot\  $:((x_0,x_1)^A,c^B)\mapsto x_c^B$, where $x_0$ and $x_1$
denote Alice's input bits, $c$ denotes Bob's input bits, and $x_c$ denotes
Bob's output bit.}~\cite{WW05b}, a noisy version of
randomized \sot\ with noise rate
$\sin^2\frac{\pi}{8}$ can also be obtained from one shared EPR-pair
while no such classical protocol exists.

The cryptographic power of quantum protocols for
two-party computations
have been investigated in ~\cite{SSS09}.
Let Alice and Bob be the two parties involve in
a two-party computation. In this model,
a primitive
is modelled by a joint probability distribution $P_{X,Y}$ 
where Alice outputs $x$
and  Bob $y$ with probability $P_{X,Y}(x,y)$. Any two-party
primitive can be randomized (the input to the functionality
are picked at random) so that its functionality is captured
by an appropriate choice of $P_{X,Y}$. We say that $P_{X,Y}$ is
trivial if it can be implemented by a correct classical protocol
against honest-but-curious (HBC) adversaries. Intuitively,
a quantum protocol for primitive $P_{X,Y}$ is correct if
once Alice and Bob get their respective outputs with
joint probability $P_{X,Y}$ then nothing else is available
to each party about the
other party's output. Such a protocol can be \emph{purified} and the measurements
 yielding the outcomes $X$ and $Y$ can be postponed to the end 
of the protocol's execution. The state of the protocol 
just before the final measurements take place, is then called 
\emph{quantum embedding} of the implemented primitive. In addition, 
\emph{regular embedding} of a primitive is defined to 
be an embedding where Alice and Bob do not posess any other (auxiliary) registers 
than the ones used to measure their respective outputs.
In~\cite{SSS09}, it is shown that although quantum protocols
can implement correctly non-trivial functionalities they will always
 leak extra information even against the  weak class
of honest-but-curious
quantum adversaries. While classical
protocols can only implement trivial primitives, quantum protocols
necessarily leak when they correctly implement something non-trivial.

In this paper, we look at another aspect of two-party quantum protocols:
their ability to compose against quantum honest-but-curious
adversaries (QHBC). In order to guarantee composability, the functionality
of a quantum protocol should be modeled by some
classical ideal functionality.
An ideal functionality is a classical description of what
the protocol achieves independently of the environment
in which it is executed. If a protocol does not admit
such a description then it can clearly not be used in
any environment while keeping its functionality,
and
such a protocol would not compose securely in all applications.
In particular, in this thesis we investigate composability 
of non-trivial quantum protocols. 
An embedding of $P_{X,Y}$ is called \emph{trivial} if both 
parties can access at least the same amount of information about 
the functionality as it is possible in some classical protocol 
for $P_{X,Y}$ in the HBC model. Otherwise, it is said to be 
\emph{non-trivial}. A quantum protocol is non-trivial
if its bipartite purification results
in a non-trivial embedding. We show that
no non-trivial 
quantum protocol composes freely
even if the adversary is restricted to be honest-but-curious.
No ideal functionality, even with an uncountable
set of rules, can fully characterize the behavior
of a quantum protocol in all environments.
This is clearly another severe limit
to the cryptographic power of two-party quantum protocols.

It is not too
difficult to show that any trivial embedding can be implemented
by  a quantum protocol that composes against QHBC adversaries.
In the other direction,
let $\ket{\psi({\pi})}\in \hil_{A}\otimes\hil_{B}$
be a non-trivial embedding of $P_{X,Y}$
corresponding to the bipartite purification of quantum protocol $\pi$.
We know that $\ket{\psi({\pi})}$ necessarily
leaks information
towards a QHBC adversary. Any ideal functionality
$\ID{{\pi}}$
for protocol $\pi$ trying to account for
honest-but-curious behaviors should allow to simulate
all measurements applied either in $\hil_{A}$ or $\hil_{B}$
through an appropriate call to $\ID{{\pi}}$. One way
to do this is to define $\ID{{\pi}}$ by a function
$[0..1]\times[0..1]
\mapsto [0..1]\times [0..1]$ where $\ID{{\pi}}(0,0)$
corresponds to the honest behavior on both sides:
$\ID{\pi}(0,0) = (x,y)$ with probability $P_{X,Y}(x,y)$
where $(x,y)$ is encoded as a pair of real numbers. Other inputs
to the ideal functionality allow for the simulation of different
strategies mounted by the QHBC adversary.
In its most general form, an ideal functionality
could have an uncountable set of possible inputs in order to
allow the simulation of all QHBC adversaries. 
We show that even allowing for these general
ideal functionalities, composed non-trivial protocols cannot be modeled
by one single ideal functionality. It means that for a protocol $\Pi$ calling 
$N$ times any non-trivial sub-protocol $\pi$ , there is a QHBC strategy that cannot be modeled 
by arbitrarily many calls of $\ID{\pi}$, each of them acting locally on a single copy of $\pi$.

In order to achieve this, we provide a generic example of such a protocol. Protocol $\Pi$
produces, as output,  a real-value ${p}$ that we call \emph{payoff}.
The payoff $p$
represents how well the adversary can compare, without error,
two factors of product states extracted from the $N$ executions of protocol
$\pi$. From a result of  \cite{KKB05},
the product states are constructed in such a way that
no individual measurement can do as well as the best coherent
measurement. It
follows that
the payoff corresponding to any adversary restricted to deal
with $\pi$ through
any ideal functionality would necessarily be lower than the one
an adversary
applying coherent strategies on both parts of the product state
could get. This implies that no ideal functionality for $\pi$ would
ever account for all QHBC strategies in $\Pi$.
Moreover, the advantage of coherent strategies over individual ones
can be made constant.
The result follows.

\section{Preliminaries}

\myparagraph{Classical Information Theory -- Dependent Part}
The following definition introduces a random variable describing the correlation between two random variables $X$ and $Y$.

\begin{definition}[Dependent part~\cite{WW04}]
\label{dep_part}
For two random variables $X,Y$, let $f_X(x) \assign P_{Y|X=x}$. Then the
\emph{dependent part of $X$ with respect to $Y$} is defined as $\dep{X}{Y} \assign f_X(X)$.
\end{definition}

The dependent part $\dep{X}{Y}$ is the minimum random variable
from the random variables computable from $X$ such that
$X\leftrightarrow \dep{X}{Y} \leftrightarrow Y$ is a Markov chain
\cite{WW04}. It means that for any random variable $K=f(X)$ such that $X\leftrightarrow K \leftrightarrow Y$ is a Markov chain, there exists a function $g$ such that $g(K)=\dep{X}{Y}$. 
Immediately from the definition we get several other properties of $\dep{X}{Y}$~\cite{WW04}:
$H(Y|\dep{X}{Y})=H(Y|X)$, $I(X;Y)=I(\dep{X}{Y};Y)$, and $\dep{X}{Y}=\dep{X}{(\dep{Y}{X})}$. The second and 
the third formula yield $I(X;Y)=I(\dep{X}{Y};\dep{Y}{X})$.

The notion of dependent part has been further investigated in
\cite{FWW04,IMNW04,WW05a}.  Wullschleger and Wolf have shown that
quantities $H(X\searrow Y|Y)$ and $H(Y\searrow X|X)$ are monotones for
two-party protocols\cite{WW05a}.  That is, none of these values can
increase during classical two-party protocols. In particular, if Alice
and Bob start without sharing any non-trivial cryptographic resource
then classical two-party protocols can only produce $(X,Y)$ such that:
$H(X\searrow Y|Y)=H(Y\searrow X|X)=0$, since $H(X\searrow Y|Y)>0$ if
and only if $H(Y\searrow X|X)>0$~\cite{WW05a}.  Conversely, any
primitive satisfying $H(X\searrow Y|Y)=H(Y\searrow X|X)=0$ can be
implemented securely in the honest-but-curious (HBC) model. We call
such primitives \emph{trivial}.

\myparagraph{Quantum Information Theory and State Distinguishability}
Let $\ket{\psi}_{AB} \in
\hil_{AB}$ be an arbitrary pure state of the joint systems $A$ and
$B$. The states of these subsystems are $\rho_A = \tr_B\proj{\psi}$
and $\rho_B=\tr_A\proj{\psi}$, respectively.  We denote by $S(A)
\assign S(\rho_A)$ and $S(B) \assign S(\rho_B)$ the von Neumann
entropy (defined as the Shannon entropy of the eigenvalues of the
density matrix) of subsystem $A$ and $B$ respectively. Since the joint
system is in a pure state, it follows easily from the Schmidt
decomposition that $S(A)=S(B)$ (see e.g.~\cite{NC00}). Analogously 
to their classical counterparts, we can define quantum conditional 
entropy $S(A|B)\assign S(AB)-S(B)$, and quantum mutual information 
$S(A;B)\assign S(A)+S(B)-S(AB)=S(A)-S(A|B).$ Even though in general, 
$S(A|B)$ can be negative, $S(A|B)\geq 0$ is always true if $A$ is a classical 
random variable. 

The following lemma gives a relation between the probability of error and the 
probability of conclusive answer of a POVM used for discriminating two 
pure state. 

\begin{lemma}[\cite{BC98}]
\label{cb98a}
Let the probability of a conclusive outcome and the error-probability of some
POVM applied to a state, sampled uniformly at random from a pair of pure states
 $(\ket{\psi_0},\ket{\psi_1})$, be denoted by 
$q_c$ and $q_{\rm{err}}$, respectively. Then
$$q_{\rm{err}}\geq \frac{1}{2}\left(q_c-\sqrt{q_c^2-(q_c-1+|\bracket{\psi_0}{\psi_1}|)^2}\right).$$
\end{lemma}

Notice that for the marginal case where $q_{\rm{err}}=0$ we get that $q_c\leq 1-|\bracket{\psi_0}{\psi_1}|$~\cite{Ivanovic87,Dieks88,Peres88}, and for the 
marginal case where $q_c=1$ (no inconclusive answer is allowed) we get 
$q_{\rm{err}}\geq \frac{1}{2}-\frac{\sqrt{1-|\bracket{\psi_0}{\psi_1}|}}{2}$~\cite{Helstrom76}.

\myparagraph{Purification}
All security questions we ask are with respect to \emph{(quantum)
  honest-but-curious} adversaries. In the classical honest-but-curious
adversary model (HBC), the parties follow the instructions of a
protocol but store all information available to them. Quantum
honest-but-curious adversaries (QHBC), on the other hand, are allowed
to behave in an arbitrary way that cannot be distinguished from their
honest behavior by the other player.

Almost all impossibility results in quantum cryptography rely upon a
quantum honest-but-curious behavior of the adversary.  This behavior
consists in {\em purifying} all actions of the honest
players. Purifying means that instead of invoking classical randomness
from a random tape, for instance, the adversary relies upon quantum
registers holding all random bits needed. The operations to be
executed from the random outcome are then performed quantumly without
fixing the random outcomes.  For example, suppose a protocol instructs
a party to pick with probability $p$ state $\ket{\phi^0}_C$ and with
probability $1-p$ state $\ket{\phi^1}_C$ before sending it to the
other party through the quantum channel $C$. The purified version of
this instruction looks as follows: Prepare a quantum register in state
$\sqrt{p}\ket{0}_R+\sqrt{1-p}\ket{1}_R$ holding the random
process. Add a new register initially in state $\ket{0}_C$ before
applying the unitary transform $U:\ket{r}_R\ket{0}_C \mapsto
\ket{r}_R\ket{\phi^r}_C$ for $r\in\{0,1\}$ and send register $C$
through the quantum channel and keep register $R$.

From the receiver's point of view, the purified behavior is
indistinguishable from the one relying upon a classical source of
randomness because in both cases, the state of register $C$ is
$\rho=p\proj{\phi^0}+(1-p)\proj{\phi^1}$. All operations invoking
classical randomness can be purified similarly\cite{LC97,Mayers97}.
The result is that measurements are postponed as much as possible and
only extract information required to run the protocol in the sense
that only when both players need to know a random outcome, the
corresponding quantum register holding the random coin will be
measured.  If both players purify their actions then the joint state
at any point during the execution will remain in pure state, until the
very last step of the protocol when the outcomes are measured.

\myparagraph{Correct Two-Party Quantum Protocols and Their Embeddings}
In this section we define when a protocol \emph{correctly implements} a 
joint distribution $P_{X,Y}$ which may correspond to some standard
cryptographic task with uniformly random inputs. We call such a 
probability distribution \emph{primitive}. As an example of a 
primitive, we can take e.g. $P_{X,Y}$ such that for all $x_0,x_1,y,c\in\{0,1\}$, 
$P_{X,Y}(x_0,x_1,c,y)=1/8$ if and only if $y=x_c$. $P_{X,Y}$ then 
corresponds to a cryptographic task known as \emph{one-out-of-two oblivious 
transfer} (\ot), first introduced by Wiesner~\cite{Wiesner83}. 
It lets Alice send two bits ($x_0,x_1$) to Bob, of which he selects one ($x_c$) to receive. 
In the randomized version, we assume the inputs $x_0$, $x_1$, and $c$ to 
be chosen uniformly at random. For standard cryptographic primitives such as 
\ot, the version with inputs can be securely implemented from the randomized 
version~\cite{WW05b}. It follows that for such primitives, considering 
the randomized version is without loss of generality. 

As a result of purification of a protocol implementing primitive $P_{X,Y}$, 
up to the point when the final measurements take place, Alice and Bob obtain a shared pure state $\ket{\psi}$. 
Without loss of generality, we may assume that the final measurements yielding 
the implemented probability distribution are in the standard (computational) basis. 
Besides the registers $A$ and $B$ needed to compute $X$ and $Y$, the players could 
use auxiliary registers $A'$ and $B'$, yielding the final state $ket{\psi}$ to be 
in $\hil_{AA'}\otimes \hil_{BB'}$, where $\hil_{AA'}$ and $\hil_{BB'}$ denote 
the subsystems controlled by Alice and Bob, respectively. Informally, 
we call $\ket{\psi}$ an \emph{embedding} of $P_{X,Y}$, if the extra working 
registers $A'$ and $B'$ do not provide any extra information to the honest players, 
measuring their respective registers $A$ and $B$ in the computational bases. By ``extra information'' 
we mean additional information about the other party's output, not available to a player from the ideal
functionality for $P_{X,Y}$. A protocol whose purification produces an embedding of $P_{X,Y}$ as the final state  is then called \emph{correct} protocol for $P_{X,Y}$. Formally, we define an embedding of and a correct protocol for a given primitive as follows:

 \begin{definition}[\cite{SSS09}]
\label{defcorrect}
\label{correct}
A protocol $\pi$ for $P_{X,Y}$ is \emph{correct} if its final state
satisfies $S(X; YB') = S(XA'; Y) =I(X;Y)$ where $X$ and $Y$ are
Alice's and Bob's honest measurement outcomes in the computational
basis and $A'$ and $B'$ denote the extra working registers of Alice
and Bob. The state $\ket{\psi}\in \hil_{AB} \otimes \hil_{A'B'}$ is
called an \emph{embedding of $P_{X,Y}$} if it can be produced by the
purification of a correct protocol for $P_{X,Y}$.
\end{definition}

Correctness is a natural restriction imposed on two-party quantum protocols, 
since nothing can prevent honest players to perform any measurement they wish 
in the systems which are not needed to compute their desired outputs.
In the following, we also use the notion of \emph{regular} embedding which, 
as it turns out, simplifies the analysis of two-party quantum protocols.

\begin{definition}[\cite{SSS09}] 
\emph{Regular} embedding of $P_{X,Y}$ is an embedding where the auxiliary registers $A'$ and $B'$ 
are trivial.
\end{definition}
\cite{SSS09} shows that any embedding of $P_{X,Y}$ can be easily converted into its regular embedding by a measurement performed on either side.

\begin{lemma}[~\cite{SSS09}]
\label{symmetry}
Let $\ket{\psi}_{AA'BB'}$ be an embedding of $P_{X,Y}$. Then $\ket{\psi}$ is locally equivalent to 
a state $\ket{\psi^*}$ in the form:
$$\ket{\psi^*}=\sum_k \lambda_k\ket{k,k}_{A'B'}\ket{\psi_k}_{AB},$$ 
where $\lambda_k$ are all nonnegative real numbers and for each $k$, $\ket{\psi_k}$ is a regular 
embedding of $P_{X,Y}$. 
\end{lemma}

It follows easily from the lemma above that Alice can convert $\ket{\psi}$ into a product state 
$\ket{\psi_k}_{AB}\otimes\ket{\varphi}_{B'}$ by a proper measurement in register $A'$. An analogous 
statement holds for Bob. 

Informally, an embedding $\ket{\psi}_{AA'BB'}$ of $P_{X,Y}$ is called \emph{trivial}, if it allows a dishonest
player to access at least the same amount of information as he/she is allowed in some classical implementation of 
$P_{X,Y}$. Formally, we define trivial and non-trivial embeddings of a given primitive as follows:

\begin{definition}[\cite{SSS09}]
Let $\ket{\psi}_{AA'BB'}$ be an embedding of $P_{X,Y}$. We call $\ket{\psi}$ a \emph{trivial} embedding of $P_{X,Y}$ if it satisfies $S(\dep{Y}{X}|AA')=0$ or $S(\dep{X}{Y}|BB')=0$. Otherwise, we call it \emph{non-trivial}. 
\end{definition} 

Notice that $P_{X,Y}$ can be implemented by the following classical protocol: 

\begin{enumerate}
\item Bob samples $x'=P_{Y|\dep{X}{Y}=x'}$ from the distribution $P_{\dep{X}{Y}}$ and sends it to Alice. He 
samples $y$ from the distribution $x'$.
\item Alice samples $x$ from the distribution $P_{X|\dep{X}{Y}=x'}$.
\end{enumerate}

Clearly, in the case where $S(\dep{X}{Y}|BB')=0$, $\ket{\psi}$ allows dishonest Bob and Alice to access at least 
as much information about the other party's outputs, as they can in the classical implementation above.  

\section{Non-Trivial Protocols and Composability}
\label{noncompose}

In the following we show that quantum protocols even characterized only by
the embeddings of the corresponding primitives (i.e. without considering
whether or not that state can be distributed fairly)
do not compose without allowing the adversary to
mount joint attacks that cannot be simulated by
attacks applied to individual copies. We are allowed 
to make this simplification because any attack of an embedding 
of a primitive can be modeled by an at least equally efficient 
(in terms of the amount of extra information accessible 
by a cheater) attack of the associated protocol. We define 
trivial protocols to be such that produce trivial embeddings.

\begin{definition}
A correct protocol for a primitive $P_{X,Y}$ is \emph{trivial}, if the embedding produced by such a protocol is trivial. Otherwise, it is called \emph{non-trivial}.
\end{definition}

In order to show 
non-composability of a non-trivial embedding $\ket{\psi}\in\hil_{ABA'B}$ of a primitive $P_{X,Y}$, satisfying t
$S_\psi(\dep{X}{Y}|BB')>0$ and $S_\psi(\dep{Y}{X}|AA')>0$, 
it is sufficient to show 
that no non-trivial regular embedding of $P_{X,Y}$  can be composed, for the following reason: Lemma~\ref{symmetry} shows that by measuring register $A'$ of $\ket{\psi}$, Alice  
converts $\ket{\psi}$ into $\ket{\psi_k}$ for some $k\in\{1,\dots,K\}$, which 
is a regular embedding of $P_{X,Y}$. If she performs such a measurement on many copies of $\ket{\psi}$, 
with high probability at least some constant fraction of them collapses into the same non-trivial regular embedding of $P_{X,Y}$. Non-composability of such a regular embedding then implies non-composability of embedding $\ket{\psi}$ of $P_{X,Y}$. 
The protocol composability questions 
can therefore be reduced to investigating composability of 
regular embeddings.  

In the following,
we formalize the weakness of non-composability 
inherent to any two-party
quantum protocol, preventing us from building
strong cryptographic primitives even from
non-trivial weak ones. This is in a sharp contrast
with quantum key distribution -- a three-party game that can be
shown to be universally composable~\cite{BHLMO05}.

Composability of quantum protocols has been studied
by Ben-Or and Mayers~\cite{BM02,BM04} and by Unruh~\cite{Unruh04}.
The former approach is an extension of Canetti's framework~\cite{Canetti01}
to the quantum case while the latter is an extension
of Backes, Pfitzmann, and Waidner~\cite{BPW04}.
We are going to consider a weaker version of
composability called {\em weak composability}
and show that almost no quantum protocol satisfies it.
Informally, we call a quantum two-party protocol weakly self-composable
if any adversarial strategy acting, possibly coherently, upon $n$ independent copies of the protocol
is equivalent to a strategy which acts individually upon each copy of the protocol.

\section{Ideal Functionalities}\label{ideal}

In order to guarantee composability, the functionality
of a quantum protocol should be modeled by some
classical ideal functionality.
An ideal functionality is a classical description of what
the protocol achieves independently of the environment
in which it is executed. If a protocol does not admit
such a description then it can clearly not be used in
any environment while keeping its functionality,
and such a protocol would not compose securely in all applications.

In the following, let $\hil_A$ and $\hil_B$ denote Alice's and Bob's quantum systems, respectively,
and let $\mathcal{X}$ and $\mathcal{Y}$ denote the set of classical outcomes of
Alice's and Bob's final measurements.

Intuitively, a pure state $\ket{\psi}\in\hil_A\otimes\hil_B$
implements the ideal functionality
\ID{\psi}\ if whatever the adversary
does on his/her part of $\ket{\psi}$, there exists
a classical input  to \ID{\psi}\
for the adversary that produces the same view.
The ideal functionality $\ID{\psi}$
accepts inputs for Alice 
and for Bob in $[0..1]$, where the elements of
$[0..1]$ encode all possible strategies
for both parties. When a party
inputs $0$ to \ID{\psi}, the outcome
of measuring this party's part of $\ket{\psi}$
in the computational basis, encoded by a number in $[0..1]$ 
is returned to the party.
This corresponds to the honest behavior. When
$m\in [0..1]$ is input to \ID{\psi}, a measurement
depending upon $m$ is applied to register $\hil_A$ (resp. $\hil_B$) of $\ket{\psi}$
and the classical outcome is returned to Alice (resp. Bob). Such a measurement acts 
only locally on the specified system.
Clearly, for $\ID{\psi}$ to be
of any cryptographic value, the set of possible
strategies should be small,
otherwise it would be very difficult to characterize
exactly what \ID{\psi}\ achieves.
As we are going to show next,
even if $\ket{\psi}$
implements such an \ID{\psi}\ where $[0..1]$ 
is used to encode
 all possible POVMs in $\hil_A$ and $\hil_B$ 
then all adversarial strategies
against $\ket{\psi}^{\otimes n}$
cannot be modeled by calls to $n$ copies of $\ID{\psi}$.

We write $\ID{\psi}(m,0)=(\tilde{w}, z)$ for the
ideal functionality corresponding to pure state
$\ket{\psi}\in \hil_A\otimes \hil_B$
with honest Bob and dishonest Alice using strategy
$m\in (0..1]$.
The output $\tilde{w}$ is provided to Alice
and $z\in [0..1]$ encoding an event in $\mathcal{Y}$ to Bob. Similarly,
we write $\ID{\psi}(0,m)=(z,\tilde{w})$
when Alice is honest and Bob is dishonest and
is using strategy $m\in (0..1]$.
Notice that an ideal functionality
for state $\ket{\psi}$ is easy to implement by letting
\ID{\psi}\  simulate Alice's and Bob's strategies
through a classical interface.

In general, \ID{\psi}\ returns one party's output as soon
as its strategy has been specified. The ideal functionality
never waits for both parties before returning the outcomes.
This models the fact that shared pure states never signal
from one party to the other.
The ideal functionality \ID{\psi}\ can be queried by one party
more than once with different strategies. The ideal
functionality keeps track of the residual state
after one strategy is applied. If a new strategy is
applied then it is applied to the residual state.
This feature captures the fact that the first measurement can be
applied before knowing how to refine it, which may happen
when Alice and Bob are involved in an interactive
protocol using only classical communication from
shared state $\ket{\psi}$.
Dishonest Alice may measure partially her part of $\ket{\psi}$
before announcing the outcome to Bob. Bob could then send
information to Alice allowing her to refine
her measurement of $\ket{\psi}$ dependently of what she received
from him. This procedure can be simulated using \ID{\psi}\ after specifying
a partial POVM for Alice's first measurement among the set
of POVMs encoded by the elements of $[0..1]$. Then, Alice
refines her first measurement by specifying
a new POVM represented by an element of  $[0..1]$ to the ideal 
functionality \ID{\psi}.

\section{Simulation}\label{simul}

A pure state $\ket{\psi}\in \hil_A\otimes\hil_B$ implements
the ideal functionality \ID{\psi}\ if any attack
implemented via POVM ${\cal M}$ by adversary Alice (resp. adversary
Bob) can be simulated by calling the ideal functionality with
some $m\in [0..1]$. The attack in the simulated
world calls \ID{\psi}\ only once as it is in the real case.
The ideal functionality \ID{\psi}\ therefore refuses to
answer more than one query per party. Remember also that
\ID{\psi}\ returns the outcome to one party as soon as
the party's strategy is specified irrespectively of whether the other
party has  specified its own.

\newcommand{\IDAM}{\ID{\mbox{\tiny{\sc amb}}}}
\newcommand{\IDOT}{\ID{\mbox{\tiny{\sc ot}}}}
\newcommand{\IDEPR}{\ID{\mbox{\tiny{\sc epr}}}}

First, let us show on an example what do we mean by simulation of an attack using the calls to
the ideal functionality.

\begin{example}\label{epr}
Consider that Alice and Bob are sharing $\ket{\Psi^{+}}=\frac{1}{\sqrt{2}}(\ket{00}+\ket{11})$
which is an embedding of the joint probability distribution $P_{X,Y}$ with $P_{X,Y}(0,0)=P_{X,Y}(1,1)=1/2$.
Alice's and Bob's honest measurement happen to be in the Schmidt basis.
We can define the ideal functionality \IDEPR\ as follows:
\[ \IDEPR(0,0) = (x,x) \mbox{ with prob. $\frac{1}{2}$.}
\]
Since both players are measuring in the Schmidt basis, it follows that
\IDEPR\ models any adversarial behavior.
\IDEPR\ is an ideal functionality for $\ket{\Psi^{+}}$
even in a context where it is a part of a larger
system.
However, $\ket{\Psi^{+}}$
is a trivial embedding!
\end{example}

Notice that any strategy against $\ket{\Psi^+}^{\otimes m}$
can be simulated by appropriate calls to $m$ copies of $\IDEPR$.
In other words, $\ket{\psi^{+}}$ is self-composable in a weak sense.
In the following section we show that in fact, all weakly self-composable
regular embeddings of joint probability distributions are trivial.

\section{Self-Composability of Embeddings}\label{composition}

We define the {\em classical weak self-composability} of a regular embedding
$\ket{\psi}\in \hil_A\otimes\hil_B$ of a joint probability distribution $P_{X,Y}$
as its ability
to be composed with itself without allowing the adversary to
get information about $X$ resp. $Y$ that is not available
through calls to independent copies of \ID{\psi}.

\begin{definition}\label{wsc}
Embedding $\ket{\psi}$ of $P_{X,Y}$ is {\em weakly self-composable} if there exists
an ideal functionality \ID{\psi}\ such that all attacks against $\ket{\psi}^{\otimes m}$
for any $m>0$ can be simulated by appropriate calls to $m$ ideal functionalities
\ID{\psi}.
\end{definition}

Next, we show that only (not necessarily all) 
trivial regular embeddings can be weakly self-composed.
The idea behind this result is the definition of a  protocol computing a function,
between Alice and Bob sharing $\ket{\psi}^{\otimes m}$ such that Bob can make the expected value of the function  strictly larger provided he has the capabilities to measure
his part of $\ket{\psi}^{\otimes m}$ coherently rather than individually. Only
individual measurements can be performed by Bob if $\ID{\psi}$ is modelling
the behavior of $\ket{\psi}$ in any situation.
Consider
that Alice and Bob are sharing a non-trivial regular embedding $\ket{\psi}$ of $P_{X,Y}$
that can be written as:
\begin{equation}
\ket{\psi} = \sum_{x\in \mathcal{X}} \sqrt{P_X(x)} \ket{x}^{A}\ket{\psi_x}^{B}. 
\end{equation}
We show in Lemma~\ref{xx} that $\ket{\psi}$ being
non-trivial (i.e. $S(X\searrow Y | \rho_B)>0$ ) implies
existence of $x_0\neq x_1\in \mathcal{X}$ such that
\begin{equation}\label{ind}
 0 < |\bracket{\psi_{x_0}}{\psi_{x_1}}|^2 < 1.
\end{equation}
Protocol~\ref{fig:challenge} challenges Bob to {\em identify}
in some sense
the state of two positions chosen uniformly and at random
among the following possibilities: 

$\{\ket{\psi_{x_0}}\ket{\psi_{x_0}}, \ket{\psi_{x_0}}\ket{\psi_{x_1}}, 
 \ket{\psi_{x_1}}\ket{\psi_{x_0}}, \ket{\psi_{x_1}}\ket{\psi_{x_1}}
\}$.
We will show that Bob, restricted to interact with his subsystem through the ideal
functionality $\ID{\psi}$, cannot make the expected value of a certain function as large as when it is allowed to interact
unconditionally (i.e. {\em coherently} )
with his subsystem.
We now prove that such $x_0,x_1\in\mathcal{X}$ exist
for any non-trivial regular embedding.
\begin{lemma}\label{xx}
If $\ket{\psi}\in \hil_A \otimes \hil_B$ is a non-trivial regular embedding of $P_{X,Y}$ then there exist
$x_0,x_1\in \mathcal{X}$ such that $\ket{\psi_{x_0}}$
and $\ket{\psi_{x_1}}$ satisfy
\[0< |\bracket{\psi_{x_0}}{\psi_{x_1}}| <1.
\]
\end{lemma}

\begin{proof}
Let us write $\ket{\psi}$ as,
\begin{equation}\label{rew1}
 \ket{\psi} = \sum_{x\in\mathcal{X}} \sqrt{P_X(x)} \ket{x}^A\ket{\psi_x}^B.
\end{equation}
Let $\{\ket{\psi^*_1},\ldots,\ket{\psi^*_\ell}\}\subseteq\{\ket{\psi_x}\}_{x\in\mathcal{X}}$
be the set of different states $\ket{\psi_x}$ available to Bob when Alice
measures $X$. Equation (\ref{rew1})
can be re-written as,
\begin{equation}\label{rew2}
\ket{\psi} = \sum_{j=1}^{\ell} \left(
\sum_{\stackrel{x\in\mathcal{X}:}{\ket{\psi_x}=\ket{\psi^*_j}}} e^{i\theta(x)}\sqrt{P_{X}(x)} \ket{x}\right) \otimes
  \ket{\psi^*_j},
\end{equation}
for some $\theta(x)\in [0\ldots 2\pi)$.

If $\{\ket{\psi^*_j}\}_{j=1}^\ell$ are mutually orthogonal then if Bob measures in this
basis no uncertainty about $X\searrow Y$ is left
contradicting the fact that $S(X\searrow Y|\rho_B)>0$.
\end{proof}

In Protocol~\ref{fig:challenge} Alice asks Bob to compare the two pure states on his side. In the next section we define a game related to the state
comparison problem and show that there is a coherent strategy which in this
game can succeed strictly better than any separable one, and therefore also
LOCC strategy on Bob's registers.

\newcommand{\challenge}{{\sc challenge}}
\begin{myfigure}{ht}
\begin{myprotocol}{\challenge}
\item Let $p\assign0$ and let Alice and Bob both know $x_0,x_1\in \mathcal{X}$ such that $0<|\bracket{\psi_{x_0}}{\psi_{x_1}}|=\tau<1$ is satisfied.
\item Alice gets $X^m = X_1,\ldots,X_m$ by measuring her part in all $m$ copies of $\ket{\psi}$ in the computational
basis.
       She identifies $4$ positions $1\leq i\neq i',j\neq j'\leq m$ such that
       $X_i=X_{i'}=x_0$ and $X_j=X_{j'}=x_1$. If such four positions do
       not exist then Alice announces to Bob that $p=0$ and aborts.
\item Alice picks $(h,h') \in \{i,i',j,j'\}$ with $h\neq h'$ such that $(X_h,X_{h'})=(\alpha,\beta)$ with probability $1/4$ for any choice of $\alpha,\beta\in\{x_0,x_1\}$ and announces $(h,h')$
       to Bob.
\item  Bob sends $b\in\{0,1,?\}$ to Alice, guessing whether the pair of pure states on the positions $h,h'$ is one of $A_0\assign\{\ket{\psi_{x_0}}\ket{\psi_{x_0}},\ket{\psi_{x_1}}\ket{\psi_{x_1}}\}$, ${\sf A}_1\assign\{\ket{\psi_{x_0}}\ket{\psi_{x_1}}, \ket{\psi_{x_1}}\ket{\psi_{x_0}}\}$, or responds by ``don't know''.
\item Alice sets the payoff value $p$: $p\assign-c$ if Bob responded incorrectly, $p\assign0$ if he answered ``don't know'', and $p\assign1$ if he answered the challenge correctly.
\end{myprotocol}
\caption{A state comparison challenge to Bob.}\label{fig:challenge}
\end{myfigure}

\section{State-Comparison Game with a Separably Inapproximable Coherent Strategy }
\label{comparison}

Consider the challenge from Protocol~\ref{fig:challenge}. In the game defined by this protocol, 
Alice lets Bob compare two states defined by a non-trivial regular embedding of a given primitive, which are either identical or different, but not orthogonal. Bob is allowed to response inconclusively however, for such 
an answer he obtains 0 points. On the other hand, if his guess is right, he obtains 1 point and if it is wrong, 
he obtains $-c$ points for some positive number $c$ which we determine later. We call his score \emph{payoff}.
With respect to the game defined by Protocol~\ref{fig:challenge}, let the maximal achievable expected payoff over the set of all measurement strategies be denoted by $p_{\max}$.
In this section we show that there exists $c$ such that the maximal average payoff $p_{\max}$
can be only achieved with a strategy coherent on the registers corresponding to the two factors of Bob's
product state. Furthermore, we show that for such a $c$ there is a constant gap between the maximal
payoff achievable with a separable strategy and $p_{\max}$. Separable measurements on a quantum
system consisting of two subsystems are such that any of their elements $M$ is in the form
$M=\sum_{i,j} F_i^0\otimes F_j^1$, where $F_i^0, F_j^1$ are the operators acting on the
respective subsystems of the given system.  According to~\cite{BDF99}, separable measurements
form a strict superset of  all LOCC measurements.

It is shown in~\cite{KKB05} that for $0<\tau<1$, the optimal no-error measurement is always coherent.
Furthermore, they prove that the highest success rate achievable by a separable unambiguous measurement
is $(1-\tau)^2$ whereas the optimal measurement has a success rate $(1-\tau)$.

Fix the value of $0<\tau<1$. For $c$ sufficiently large the best coherent strategy is to apply the best unambiguous measurement with the correct-answer rate $1-\tau$,
and to output \emph{don't know} for an uncertain result. Therefore, for some $c$ we have
$p_{\max}=1-\tau$. Let $p_s$ denote the supremum of average payoffs in the game from  Protocol~\ref{fig:challenge} achievable by separable strategies.

\begin{theorem}
\label{game}
In the game from Protocol~\ref{fig:challenge} there exists $c>0$ such that $p_s\leq p_{\max}-f(\tau)$, where $f(\tau)>0$ whenever
$0<\tau<1$.
\end{theorem}

Before proving the actual theorem, we introduce a useful lemma.

\begin{lemma}
\label{cb98}
Let $\ket{\varphi_0},\ket{\varphi_1}\in\mathcal{H}$ be pure states such that 
$|\bracket{\varphi_0}{\varphi_1}|=\tau$. For a discrimination strategy 
$\mathcal{S}$ with three possible outcomes $0$, $1$, and $``don't\ know"$, 
let $q_c$ denote the probability of a conclusive answer and $q_{\rm{err}}$ the probability of a wrong answer. Then,
$$q_c\leq 2q_{\rm{err}}+1-\tau+2\sqrt{q_{\rm{err}}(1-\tau)}.$$
\end{lemma}

\begin{proof}
According to Lemma~\ref{cb98a},
$$q_{\rm{err}}\geq \frac{1}{2}\left(q_c-\sqrt{q_c^2-(q_c-(1-\cos\theta))^2}\right).$$
Equivalently, we get:
$$\sqrt{q_c^2-(q_c-(1-\cos\theta))^2}\geq q_c-2q_{\rm{err}}.$$
By squaring both sides of the inequality we obtain:
\begin{eqnarray}
2q_c(1-\cos\theta)-(1-\cos\theta)^2&\geq& q_c^2+4q^2_{\rm{err}}-4q_cq_{\rm{err}}\nonumber\\
q_c^2-q_c(4q_{\rm{err}}+2(1-\tau))+(1-\tau)^2+2q^2_{\rm{err}}&\leq& 0\label{qcineq}.
\end{eqnarray}
By solving the quadratic equation
$$q_c^2-q_c(4q_{\rm{err}}+2(1-\tau))+(1-\tau)^2+2q^2_{\rm{err}}= 0,$$
we get the solutions $2q_{\rm{err}}+1-\tau\pm 2\sqrt{q_{\rm{err}}(1-\tau)},$ implying the solutions of~(\ref{qcineq}) to be 
$$q_c\leq 2q_{\rm{err}}+1-\tau+2\sqrt{q_{\rm{err}}(1-\tau)}.$$
\end{proof}

\begin{proof}[\ Theorem\ref{game}]
The method we use is the following: For given parameters $\tau,c\in\mathbb{R}$ such that $0<\tau<1$ and $c>0$, and an additional parameter
$k>0$, we divide the set of all separable measurements into three subsets according to the
probability $q_{\rm err}$ of Bob's incorrect (conclusive) answer in the state-comparison, expressed as a function of $c$, $k$,
and $\tau$. We construct an upper bound on $p_s$ in each of the three sets separately
and dependently on $c$, $k$, and $\tau$. Finally, we find the conditions for $c$ and $k$ such
that in all three sets we get $p_s\leq p_{\max}-f(\tau)$ for some $f(\tau)>0$.

\cite{KKB05} shows that the best separable unambiguous strategy for solving the 2-out-of-2 state
comparison problem is applying the best unambiguous measurements on each part of Bob's register
independently. Lemma~\ref{sepapprox} (see Appendix~\ref{app:sep_approx}) says that the payoff achieved by such a strategy in the case where
probability $q_{\rm err}$ is small, is close to the optimal payoff.
The analysis of such a situation is captured in the first of the three cases,
where we consider the separable measurements with $q_{\rm err}\leq \frac{1}{2k(c+1)}$. 

\textbf{1.} ($q_{\rm err}\leq \frac{1}{2k(c+1)}$) Lemma~\ref{sepapprox} shows that to any separable measurement $\mathcal{M}=(E_0,E_1,E_?)$
with probability of error $q_{\rm err}\leq \frac{1}{2(c+1)k}$ and the expected payoff $p$, there
exists a separable measurement $\mathcal{M}'=(E'_0,E'_1,E'_?)$ with the expected payoff $p'$,
satisfying $p\leq p'+\frac{1}{k}+O(1/\sqrt{c})$, such that its elements can be
written in the form:
$$E'_0= G^0_0\otimes G^1_0+G^0_1\otimes G^1_1,\ \ \ E'_1= G^0_0\otimes G^1_1+G^0_1\otimes G^1_0, \ \ \ 
E'_?= 1-E'_0-E'_1,$$
where the upper index of $G_\alpha^\beta$ refers to the subsystem and the lower index determines the guess of the state of the corresponding subsystem.

The upper bound on the value of $p'$ which we compute next, can then be used to upper bound $p$.
Consider an extended problem where Bob is supposed to identify each factor
of his product state (in contrast to just comparing the factors in the game). Let $q_{\rm err}^0$, $q_{\rm err}^1$, and $q_c^0$, $q_c^1$ denote the
probabilities of Bob's incorrect resp. conclusive answers in each of his subsystems. Then the
probability of comparing the states incorrectly can be expressed as follows:
$$q_{\rm err}=q_{\rm err}^0(q_c^1-q_{\rm err}^1)+q_{\rm err}^1(q_c^0-q_{\rm err}^1)
=q_c^1q_{\rm err}^0+q_c^1q_{\rm err}^1-2q_{\rm err}^0q_{\rm err}^1.$$
For separable strategies for which $q_c^1<1-\tau-2/k$ or $q_c^0<1-\tau-2/k$, we obtain
$p'<1-\tau-2/k$ and hence, $p<1-\tau-1/k+O(1/\sqrt{c})$ due to Lemma~\ref{sepapprox}.
For $c$ sufficiently large
we then get:
\begin{equation}
\label{1case}
p_{\max}-p\geq \frac{1}{2k}.
\end{equation}

Next, we discuss the case (not disjoint with the previous one) where both $q_c^0,q_c^1\geq 1-\tau-1/k=: \gamma$, which implies
that
\begin{equation}
\label{errlowbound}
q_{\rm err}\geq \gamma(q_{\rm err}^0+q_{\rm err}^1)-2q_{\rm err}^0q_{\rm err}^1.
\end{equation}
For upper bounding
the probability $q_c^0$ of a conclusive answer of the measurement $\mathcal{M}'$ we use Lemma~\ref{cb98} (an analogous formula holds for $q_c^1$):
$$q_c^0\leq 2q_{\rm err}^0+1-\tau+2\sqrt{q_{\rm err}^0(1-\tau)}.$$
The probability of correct state-identification in the first of Bob's subsystems then satisfies:
\begin{equation}
\label{okupbound}
q_c^0-q_{\rm err}^0\leq q_{\rm err}^0+1-\tau+2\sqrt{q_{\rm err}^0(1-\tau)}.
\end{equation}

Inequalities~(\ref{errlowbound}) and~(\ref{okupbound})
give us an upper bound on $p'$ for $c>9$:
\begin{eqnarray*}
p'&\leq &-cq_{\rm err}+q_{\rm err}^0q_{\rm err}^1\\
&+&
(q_{\rm err}^0+1-\tau+2\sqrt{q_{\rm err}^0(1-\tau)})
(q_{\rm err}^1+1-\tau+2\sqrt{q_{\rm err}^1(1-\tau)})\\
&\leq &-cq_{\rm err}+(1-\tau)^2+2(\sqrt{q_{\rm err}^0}+\sqrt{q_{\rm err}^1})+9q_{\rm err}\\
&\leq &(1-\tau)^2+\frac{4\sqrt{q_{\rm err}}}{\sqrt{\gamma}}\leq (1-\tau)^2+\frac{4}{\sqrt{2\gamma k(c+1)}}
\end{eqnarray*}
hence by Lemma~\ref{sepapprox}, $p\leq (1-\tau)^2+\frac{4}{\sqrt{2\gamma k(c+1)}}+\frac{1}{k}+O(1/\sqrt{c}).$
For $c$ sufficiently large we get:
\begin{equation}
\label{1caseb}
p\leq (1-\tau)^2+\frac{2}{k}.
\end{equation}

\textbf{2.} Second, we assume that $\frac{1}{2k(c+1)}< q_{\rm err}\leq \frac{1}{256(1-\tau)}$.
To upper bound the probability of comparing the states correctly, we use the same argument as in~(\ref{okupbound}) and get that:
$$q_c-q_{\rm err}\leq q_{\rm err}+1-\tau+2\sqrt{q_{\rm err}(1-\tau)},$$
where $q_c$ denotes the probability of a conclusive outcome.
This inequality implies the upper bound on $p$:
$$p\leq -cq_{\rm err}+(q_c-q_{\rm err})\leq -\frac{c-1}{c+1}\cdot\frac{1}{2k}+1-\tau+2\sqrt{q_{\rm err}(1-\tau)},$$
yielding that for $c$ sufficiently large,
\begin{equation}
\label{2case}
p\leq -\frac{1}{2k}+1-\tau+2\sqrt{q_{\rm err}(1-\tau)}.
\end{equation}
Consequently, we have three upper bounds on the value of $p$, given by~(\ref{1case}), (\ref{1caseb}), and (\ref{2case}): $B_0\assign 1-\tau-\frac{1}{2k}$, $B_1\assign (1-\tau)^2+\frac{2}{k}$, and $B_2\assign1-\tau+2\sqrt{q_{\rm err}(1-\tau)}-\frac{1}{2k}$. Since $B_2\geq B_0$, we only have to find $f(\tau)$ and $k$ such that
$B_1,B_2\leq (1-\tau)-f(\tau)$, or equivalently:
\begin{eqnarray*}
2f(\tau)+4\sqrt{q_{\rm err}(1-\tau)}&\leq &\frac{1}{k}\leq  \frac{\tau(1-\tau)-f(\tau)}{2}\\  
\frac{5}{2}f(\tau)&\leq &\frac{\tau(1-\tau)}{2}-4\sqrt{q_{\rm err}(1-\tau)}.
\end{eqnarray*}
It is easy to verify that for $d\leq \frac{1}{256(1-\tau)}$, the two inequalities are satisfied for $k\assign\frac{20}{9\tau(1-\tau)}$ and $f(\tau)\assign\frac{\tau(1-\tau)}{10}$. Thus, there exists $c>0$ such that in any separable strategy with the probability of error $q_{\rm err}\leq \frac{1}{256(1-\tau)}$ and the expected payoff $p$:
$$p\leq p_{\max}-\frac{\tau(1-\tau)}{10}.$$

\textbf{3.} For separable strategies with the probability of error $q_{\rm err}>\frac{1}{256(1-\tau)}$, we can simply set $c>256(1-\tau)$ which ensures that the payoff $p\leq 0$.

Set $c$ to be the maximum over the values required by the discussed subcases. For such a $c$ and any separable strategy, the corresponding expected payoff $p$ satisfies $p\leq p_{\max}-\frac{\tau(1-\tau)}{10}$, yielding that
$$p_{\max}-p_s\geq \frac{\tau(1-\tau)}{10}.$$
\end{proof}

\section{Only Trivial Embeddings Can Be Composed}
\label{final}
As a straightforward corollary of Theorem~\ref{game}, we now get that there exists a constant $c$ such that
any Bob restricted to interact with his system
through the ideal functionality $\ID{\psi}^{\otimes m}$ can never get the
expected value of $p$ as large and not even close as with the best coherent strategy.
This remains true for any possible description of the ideal
functionality since even if $\ID{\psi}$  allowed to specify an arbitrary
POVM then the ideal functionality would not be as good as the best
coherent strategy.

Notice that any strategy Bob may use for querying the ideal
functionality $\ID{\psi}$ for both systems involved in order to
pass the challenge with success, can also be carried by two
parties restricted to local quantum operation and classical
communication (LOCC). This is because $\ID{\psi}$
only returns classical information. Local quantum operations
can be performed by asking $\ID{\psi}$ to apply a POVM
to a local part of $\ket{\psi}$.

We now formally prove that non-trivial regular embeddings do not compose since Bob
can always succeed better in Protocol~\ref{fig:challenge} if
he could measure all his registers coherently.

\begin{theorem}\label{thsimul}
Only trivial regular embeddings of a primitive $P_{X,Y}$ are weakly self-composable.
\end{theorem}

\begin{proof}
Let $\ket{\psi}=\sum_{x\in\mathcal{X}}\sqrt{P_X(x)}\ket{x}\ket{\psi_x}$ be a non-trivial regular embedding of $P_{X,Y}$.
According to Lemma~\ref{xx} there exist $x_0,x_1\in\mathcal{X}$ such that
$0<|\bracket{\psi_{x_0}}{\psi_{x_1}}|<1$. 
Theorem~\ref{game} then
implies that there is $c\in \mathbb{R}^+$ such that in Protocol~\ref{fig:challenge} played with $\ket{\psi_{x_0}}$ and $\ket{\psi_{x_1}}$
satisfying the condition above, the expected payoff achievable by the best coherent
strategy is strictly better than what can be achieved by separable i.e. also LOCC strategies. By definition
of weak self-composability it means that non-trivial regular embedding $\ket{\psi}$ of $P_{X,Y}$ is not
weakly self-composable.
\end{proof}

\begin{corollary}
Only trivial (correct) two-party quantum protocols are weakly self-composable.
\end{corollary}

\begin{proof}
The statement follows from the fact that any quantum honest-but-curious attack 
of an embedding can be modeled by an attack of the corresponding two-party protocol. 
Lemma~\ref{symmetry} shows that for any party 
there is a measurement converting a regular embedding $\ket{\psi}\in\hil_{ABA'B'}$ of a primitive 
$P_{X,Y}$ into an embedding $\ket{\psi_k}$ of $P_{X,Y}$ for some $k\in\{1,\dots,K\}$. 
The other party can also learn the index $k$ by measuring his/her additional register.
Non-composability of 
non-trivial quantum two-party protocols for $P_{X,Y}$ then follows from non-composability of 
non-trivial regular embeddings of $P_{X,Y}$ by including a pre-stage into the game from 
Protocol~\ref{fig:challenge}. In this stage, Alice and Bob convert each of the
many embeddings of $P_{X,Y}$ corresponding to the protocol copies into a regular embedding of $P_{X,Y}$
known to both parties.  
This conversion results into a non-trivial 
regular embedding of $P_{X,Y}$ with constant probability. This is because if all regular embeddings in 
the conversion-range were trivial, then the measurement converting the 
embedding into regular embeddings could be used as a part of a measurement 
revealing $\dep{X}{Y}$ completely to Bob, or revealing $\dep{Y}{X}$ completely 
to Alice. Hence, such an embedding and the corresponding protocol would then be trivial. 
Due to the law of 
large numbers, from several copies of 
an embedding Alice obtains at least some constant fraction of the same 
non-trivial regular embeddings except of probability negligible in the number of 
copies. Alice and Bob then play the game from Protocol~\ref{fig:challenge}, 
using the subset of copies where Alice obtained the same non-trivial 
regular embedding. 
\end{proof}

Finally, let us mention several facts related particularly to (non-)composability 
of trivial two-party quantum protocols implementing trivial primitives. Clearly, every trivial primitive 
has a protocol which is composable against quantum honest-but-curious 
adversaries, namely the classical one implementing the primitive 
securely in the HBC model. Formally, for a trivial $P_{X,Y}$
we show composability 
of quantum protocols implementing only  
$P_{\dep{X}{Y},\dep{Y}{X}}$ (which corresponds to secure implementation 
in the HBC model) instead of $P_{X,Y}$, where the desired 
distribution $P_{X,Y}$ is obtained from the implementation of 
$P_{\dep{X}{Y},\dep{Y}{X}}$ by local randomization. Since a trivial 
primitive satisfies $H(\dep{X}{Y}|\dep{Y}{X})=H(\dep{Y}{X}|\dep{X}{Y})=0$ 
or in other words, the  
implemented dependent parts are accessible to both parties already 
in one protocol copy, coherent attacks do not help in getting any 
more information. Because the rest of $X$ and $Y$ is computed 
purely locally, there is no attack, individual or coherent, revealing any information about 
the result of this operation.   

On the other hand, not all 
protocols for trivial primitives are composable. As an example let us take a protocol for a
primitive $P_{X,Y}$ defined by $P_{X,Y}(0,0)=P_{X,Y}(1,0)=3/8$, 
$P_{X,Y}(0,1)=P_{X,Y}(1,1)=1/8$, represented by the following 
regular embedding:
$$\ket{\psi}=\frac{1}{\sqrt{2}}\ket{0}\otimes \left(\frac{\sqrt{3}}{2}\ket{0}+\frac{1}{2}\ket{1}\right)+\frac{1}{\sqrt{2}}\ket{1}\otimes \left(\frac{\sqrt{3}}{2}\ket{0}-\frac{1}{2}\ket{1}\right).$$ 
Such an embedding (and therefore, the corresponding protocol) is trivial because it implements a trivial primitive. 
Formally,  $0=H(\dep{X}{Y}|Y)$ and  $H(\dep{X}{Y}|Y)\geq S(\dep{X}{Y}|B)$ 
imply that $S(\dep{X}{Y}|B)=0$.   
On the other hand, the states 
$$\ket{\psi_0}\assign \frac{\sqrt{3}}{2}\ket{0}+\frac{1}{2}\ket{1},\ \ \ \ 
\ket{\psi_1}\assign \frac{\sqrt{3}}{2}\ket{0}-\frac{1}{2}\ket{1}$$
that Bob gets for Alice's respective outcomes $0$ and $1$ of the measurement 
in the canonical basis, satisfy the condition $0<|\bracket{\psi_0}{\psi_1}|<1$ 
from Protocol~\ref{fig:challenge}. Hence, the arguments from the proof of 
Theorem~\ref{game} apply, yielding that $\ket{\psi}$ cannot be composed. 

\bibliographystyle{alpha}
\bibliography{crypto,qip,procs}

\appendix
\section{Lemma~\ref{sepapprox} from the proof of Theorem~\ref{game}}
\label{app:sep_approx}
Before starting with the actual Lemma~\ref{sepapprox}, we formulate and prove
an auxiliary lemma, needed for the main proof. In the following, 
$\|T\|_\infty$ denotes the norm of an operator $T\in\mathbb{C}^{n\times n}$, which equals the operator's
largest singular value.

\begin{lemma}
\label{lemma_nonlocal}
Let $f:\mathbb{R}^+\rightarrow \mathbb{C}^{2\times 2}$  be a function mapping $c$ positive into
a positive-semidefinite operator $F_c\in\mathbb{C}^{2\times 2}$ such that $\|F_c\|_\infty=1$
and for some unit vector $\ket{v_0}\in \mathbb{C}^2$, $\bra{v_0}F_c\ket{v_0}\in O(1/c)$.
Then the dominant eigenvector of $F_c$ is of the form
$\gamma^c_0\ket{v_0}+\gamma^c_1\ket{v_1}$, where $\bracket{v_0}{v_1}=0$,
$|\gamma^c_0|^2+|\gamma^c_1|^2=1$, and $|\gamma^c_0|^2\in O(1/c)$. Furthermore, the second largest
eigenvalue $\lambda_c$ of $F_c$ satisfies $\lambda_c \in O(1/c)$.
\end{lemma}

\begin{proof}
Let us write $F_c$ in the form:
$F_c=M_c^\dag M_c$, for a matrix $M_c\in\mathbb{C}^{2\times 2}$.
This is possible due to the fact that $F_c$ is positive-semidefinite. We define
$\ket{u_0}\assign M_c\ket{v_0}$ and $\ket{u_1}\assign M_c\ket{v_1}$. According to
the assumption, $\bracket{u_0}{u_0}\in O(1/c)$. Let us write the (unit) dominant eigenvector
of $F_c$ in the basis $\{\ket{v_0},\ket{v_1}\}$ as:
$$\ket{w}=\gamma^c_0\ket{v_0}+\gamma^c_1\ket{v_1}.$$
It follows that
$$1=\bra{w}F_c\ket{w}=|\gamma^c_0|^2\bracket{u_0}{u_0}+|\gamma^c_1|^2\bracket{u_1}{u_1}+2{\rm Re}(\overline{\gamma^c_0}\gamma^c_1
\bracket{u_0}{u_1}).$$

Assume that there exists an unbounded increasing sequence of positive numbers such that for its elements $c$, we get
$|\gamma^c_1|^2=1-\Theta(1/c^{\delta})$ for $1/2\leq \delta<1$. From $\bracket{u_0}{u_0}\in O(1/c)$
we get that $|\bar{\gamma^c_0}\gamma^c_1\bracket{u_0}{u_1}|\in \Theta(1/c^\delta)$ and $|\bracket{u_0}{u_1}|\in O(1/\sqrt{c})$,
yielding that
$$|\gamma^c_0|\in \omega(1/c^{\delta-1/2}).$$
Since $\ket{w}$ is a unit vector, for some $k$ positive, we get
$$1=|\gamma^c_1|^2+|\gamma^c_0|^2\geq 1-\frac{k}{c^\delta}+|\gamma^c_0|^2,$$
and thus, $|\gamma^c_0|^2\in O(1/c^\delta).$ From the two conditions we conclude 
that
$$|\gamma^c_0|^2\in \omega(1/c^{2\delta-1})\cap O(1/c^\delta),$$
yielding that $\delta=1$, since the intersection of the two sets
has to be non-empty. Therefore, the function $f$ satisfies
\begin{equation}
\label{gamma0}
|\gamma^c_0|^2\in O(1/c)
\end{equation}
on the entire domain.

Now we upper bound the second largest eigenvalue of $F_c$. Since the second eigenvector $\ket{w^\bot}$ of $F_c$
is orthogonal to its dominant eigenvector, it can be written in the form:
$$\ket{w^\bot}=\tilde{\gamma}^c_1\ket{v_0}+\tilde{\gamma}^c_0\ket{v_1},$$
where $|\tilde{\gamma}^c_1|=|\gamma^c_1|$ and $|\tilde{\gamma}^c_0|=|\gamma^c_0|$.
We get that
$$\lambda_c=\bra{w^\bot}F_c\ket{w^\bot}=|\tilde{\gamma}^c_1|^2\bracket{u_0}{u_0}+|\tilde{\gamma}^c_0|^2\bracket{u_1}{u_1}+2{\rm Re}
(\overline{\tilde{\gamma^c_1}}\tilde{\gamma^c_0}\bracket{u_0}{u_1}).$$
From the assumption $\bracket{u_0}{u_0}\in O(1/c)$ and~(\ref{gamma0}) we conclude that
$$\lambda_c\in O(1/c).$$
\end{proof}

\begin{lemma}
\label{sepapprox}
Let $c,k>0$. Consider the game from Prot.~\ref{fig:challenge} and let $X$ and $Y$ denote the respective registers of Bob, corresponding to Alice's choices of $h$ and $h'$. To any strategy based on the outcomes of a separable measurement $\mathcal{M}=(E_0,E_1,E_?)$ on $\hil_{X}\otimes\hil_{Y}$ with probability of error
$q_{\rm err}\leq \frac{1}{2(c+1)k}$ and the expected payoff $p$, there exists a strategy using a separable
measurement $\mathcal{M}'=(E'_0,E'_1,E'_?)$ in the form:
$$E'_0=G^0_0\otimes G^1_0+G^0_1\otimes G^1_1,\ \ \ E'_1=G^0_0\otimes G^1_1+G^0_1\otimes G^1_0, \ \ \ E'_?=1-E'_0-E'_1$$
with the expected payoff $p'$, satisfying:
$$|p-p'|\in \frac{1}{k}+O(1/\sqrt{c}).$$
\end{lemma}

\begin{proof}
For simplicity of the notation, let us define $\ket{\psi_0}\assign \ket{\psi_{x_0}}$ and $\ket{\psi_1}\assign \ket{\psi_{x_1}}$, where $\ket{\psi_{x_0}}$ and $\ket{\psi_{x_1}}$ come from Prot.~\ref{fig:challenge}.

Every element of a separable measurement on $\hil_{X}\otimes\hil_{Y}$
can be written as a sum of tensor products of positive semi-definite
operators. In particular, the elements of $\mathcal{M}$ can
be written in the form:
$$E_{b(x,y)}\assign\sum_{x,y} F^0_{b(x,y),x}\otimes F^1_{b(x,y),y}.$$
Operators $F^0_{b(x,y),x}\otimes F^1_{b(x,y),y}$ can be viewed as the
elements of a new measurement $\mathcal{N}$, refining $\mathcal{M}$.
Since the states $\ket{\psi_{0}}$ and $\ket{\psi_{1}}$ span a 2-dimensional Hilbert space, all operators $F^0_{b_{x,y},x}$ and $F^1_{b_{x,y},y}$ can be restricted to correspond to  $2\times 2$ matrices in some basis of this space. 

The function $b:(x,y)\rightarrow \{0,1,?\}$ is a post-processing function
of the outcomes of $\mathcal{N}$, determining the outcome of $\mathcal{M}$
(0 corresponds to the states being equal, 1 to them being different, and ? denotes an inconclusive answer).
Let $A$ denote the sets of all pairs $(x,y)$ of outcomes of $\mathcal{N}$.
To every pair $(x,y)\in A$ we assign
$(q^0_x,q^1_y)\in[0,1/2]^2$ -- the probabilities of error in guessing the factor states of $\hil_{X}$ and $\hil_{Y}$,
conditioned on measuring $x$ and $y$, respectively. Let $W_0$ and $W_1$ denote
the random variables assigned to the states of $\hil_{X}$ and $\hil_{Y}$, respectively. The probability space of both $W_0$ and $W_1$ is $\{0,1\}$, since the state of either of the subsystems is  $\ket{\psi_0}$ or $\ket{\psi_1}$.
For $\zeta\in\{0,1\}$, let
$x\rightarrow \zeta$, $y\rightarrow \zeta$
stand for ${\rm Pr}[W_0=1-\zeta|x],{\rm Pr}[W_1=1-\zeta|y]\leq \frac{1}{2(c+1)}$, respectively,
where the probabilities are conditioned on the outcomes of $\mathcal{N}$ in the respective subsystems.
Consider 
measurement $\mathcal{M}^*\assign(E^*_0,E^*_1,E^*_?)$ with the same refined set of outputs $A$
as $\mathcal{M}$ (which now will be indexed differently) in the following form:

$$
E^*_0= E^*_{0,0}+E^*_{1,1},\ \ \ \ \ 
E^*_1= E^*_{0,1}+E^*_{1,0}, \ \ \ \ \ 
E^*_?= \I-E^*_0-E^*_1,
$$
where
\begin{equation}
\label{Estar}
E^*_{\alpha,\beta} \assign \sum_{x\rightarrow \alpha,y\rightarrow \beta} F^0_{\alpha,x}\otimes F^1_{\beta,y}.
\end{equation}
We show that the difference of the expected payoff $p$ of $\mathcal{M}$ and the expected payoff $p^*$ of $\mathcal{M}^*$ satisfies:
\begin{equation}
\label{differ1k}
|p-p^*|\leq \frac{1}{k}.
\end{equation}

Since the refined sets of possible outcomes of both $\mathcal{M}^*$ and $\mathcal{M}$ are the same, the two measurements only
differ in the post-processing functions denoted by $b$ and $b^*$, respectively. In other words, 
$\mathcal{M}^*$ differs from $\mathcal{M}$ in the  arrangement of the same set of summands in the three sums defining measurement elements $(E_0,E_1,E_?)$ and $(E^*_0,E^*_1,E^*_?)$.

 Consider any strategy which upon measuring $(x,y)$ yields a conclusive answer.
For the corresponding expected payoff $p_{x,y}$ conditioned on measuring
$(x,y)$ we then get:
\begin{eqnarray}
p_{x,y}&=&(1-q^0_x)(1-q^1_y)+q^0_xq^1_y-c(q^0_x(1-q^1_y)+(1-q^0_x)q^1_y)\nonumber\\
&=&1-(c+1)(q^0_x+q^1_y-2q^0_xq^1_y)\label{payoff_xy}.
\end{eqnarray}

If on the other hand, measuring $(x,y)$ implies the answer of $\mathcal{M}$ to be inconclusive,
the expected payoff conditioned on measuring $(x,y)$ will be $0$. Consequently,
the optimal post-processing strategy (with the maximum payoff) should output $b(x,y)=?$ for every $(x,y)$ satisfying
$q^0_x+q^1_y-2q^0_xq^1_y>\frac{1}{c+1}$, otherwise it outputs a conclusive answer. In particular, the output
should be inconclusive for all pairs $(x,y)$ such that $q^0_x>\frac{1}{c+1}$ or $q^1_y>\frac{1}{c+1}$, and conclusive
if both $q^0_x,q^1_y\leq \frac{1}{2(c+1)}$.

However, only the knowledge that
$(q^0_x,q^1_y)\in[0,\frac{1}{c+1}]^2\setminus [0,\frac{1}{2(c+1)}]^2$ does not allow us to determine what
is the best output in order to maximize the payoff. We analyze this problem with respect
to the probability of error allowed for the post-processing function.

We assume that the answer of $\mathcal{N}$ with the post-processing function $b$ can be false with probability at most $q_{\rm err}\leq \frac{1}{2k(c+1)}$.
According to Markov's inequality, measuring $(x,y)$ such that either $q^0_x>kq_{\rm err}$ or
$q^1_y>kq_{\rm err}$ does not allow to output a conclusive answer with probability larger than $1/k$.
Thus, for either $q^0_x>\frac{1}{2(c+1)}$ or $q^1_y>\frac{1}{2(c+1)}$, the answer cannot be conclusive with probability
larger than $1/k$. In the latter we analyze the difference of the expected payoffs for the post-processing function $b$
 and for a newly defined $b^*$ such that for any $(x,y)$ satisfying $q^0_x>\frac{1}{2(c+1)}$ or $q^1_y>\frac{1}{2(c+1)}$,
the output is $b^*(x,y)=?$.

Consider every pair $(x,y)$ such that by modifying $b(x,y)$ into $b^*(x,y)$, $p_{x,y}$ decreases and
compute the difference of $p_{x,y}$ and $p^*_{xy}$ in this case. We have that either
$q^0_x\in(\frac{1}{2(c+1)},\frac{1}{c+1}]$ or $q^1_y\in (\frac{1}{2(c+1)},\frac{1}{c+1}]$, yielding that
$$p_{x,y}=1-(c+1)(q^0_x+q^1_y-2q^0_xq^1_y)<\frac{1}{2}.$$
It means that for every pair $(x,y)$ for which the value of the post-processing function was modified, $p_{x,y}$
decreased by at most $1/2$. However, since the answer of $\mathcal{M}$ is false with probability at most $q_{\rm err}$,
the functions $b$ and $b^*$ cannot differ anywhere except for a set of $(x,y)$ measured with probability
at most $1/k$, concerning that $q_{\rm err}\leq\frac{1}{2k(c+1)}$. This gives us\begin{equation}
\label{payoffdif}
|p-p^*|\leq\frac{1}{k}.
\end{equation}

We have shown that a separable measurement $\mathcal{M}$ can be approximated by a separable measurement
$\mathcal{M}^*$ in the special form. In the following we show that $\mathcal{M}^*$ can be approximated by a
measurement in the form from the statement up to a difference in payoffs which is
in $O(1/\sqrt{c})$. The statement of the lemma then follows from the triangle inequality.

Our next goal is to construct a measurement $\mathcal{M}'=(E'_0,E'_1,E'_?)$ in the form:
$$
E'_0=G^0_{00}\otimes G^1_{00}+G^0_{11}\otimes G^1_{11},\ \ \ \
E'_1=G^0_{01}\otimes G^1_{01}+G^0_{10}\otimes G^1_{10},\ \ \ \ 
E'_?=\I-E'_0-E'_1,
$$
approximating the measurement $\mathcal{M}^*$ with respect to the expected payoff.
In the definition of the elements of $\mathcal{M}'$,
the upper index of $G_{ab}^\zeta$ specifies the subsystem, the first bit of the lower index
determines the outcome in the first subsystem, and the second bit of the lower index determines
the outcome in the second subsystem.

Consider the previously constructed measurement $\mathcal{M}^*$. Fix $\alpha,\beta\in\{0,1\}$ and define
$F_x^0\assign \frac{F_{\alpha,x}^0}{\|F_{\alpha,x}^0\|_\infty}$, $F_y^1\assign \frac{F_{\beta,y}^1}{\|F_{\beta,y}^1\|_\infty}$,
$\mu_{x,y}\assign \|F_{\alpha,x}^0\|_\infty\cdot \|F_{\beta,y}^1\|_\infty.$
First, we construct positive-semidefinite operators $\tilde{G}_{\alpha,\beta}^0\otimes \tilde{G}_{\alpha,\beta}^1$, approximating
$$E^*_{\alpha,\beta}=\sum_{x\rightarrow \alpha,y\rightarrow \beta}\mu_{x,y}F_x^0\otimes F_y^1$$
 (defined by~(\ref{Estar})), where the guesses of $\alpha$ and $\beta$ conditioned on
measuring $F_x^0$ and $F_y^1$ are incorrect with probability at most $\frac{1}{2(c+1)}$. We require these operators to satisfy:
\begin{enumerate}
\item $p_{E^*_{\alpha,\beta}}=p_{\tilde{G}_{\alpha,\beta}^0\otimes \tilde{G}_{\alpha,\beta}^1},$
where  $p_{E^*_{\alpha,\beta}}$ and $p_{\tilde{G}_{\alpha,\beta}^0\otimes \tilde{G}_{\alpha,\beta}^1}$ denote the expected payoffs
conditioned on measuring $E^*_{\alpha,\beta}$ and $\tilde{G}_{\alpha,\beta}^0\otimes \tilde{G}_{\alpha,\beta}^1$, respectively.
\item
For all $\zeta_0,\zeta_1,\alpha,\beta \in\{0,1\}:$
$$ \left|\bra{\psi_{\zeta_0},\psi_{\zeta_1}}\tilde{G}_{\alpha,\beta}^0\otimes 
\tilde{G}_{\alpha,\beta}^1\ket{\psi_{\zeta_0},\psi_{\zeta_1}}-\bra{\psi_{\zeta_0},\psi_{\zeta_1}}E^*_{\alpha,\beta}\ket{\psi_{\zeta_0},\psi_{\zeta_1}}\right|\in O(1/\sqrt{c}),$$
\item
$\|\sum_{\alpha,\beta} \tilde{G}_{\alpha,\beta}^0\otimes\tilde{G}_{\alpha,\beta}^1\|_\infty\in 1+O(1/c).$
\end{enumerate}

We now describe the construction of operators $\tilde{G}_{\alpha,\beta}^0$ and $\tilde{G}_{\alpha,\beta}^1$.
The respective dominant eigenvectors of $F_x^0$ and $F_y^1$ can be written as
\begin{eqnarray*} 
\ket{w_0}&=&\gamma_{0,x}^0\ket{\psi_{1-\alpha}}+\gamma_{1,x}^0\ket{\psi_{1-\alpha}^\bot}, \\
\ket{w_1}&=&\gamma_{0,y}^1\ket{\psi_{1-\beta}}+\gamma_{1,y}^1\ket{\psi_{1-\beta}^\bot},
\end{eqnarray*}
where for each $\zeta\in\{0,1\}$, $\ket{\psi_\zeta^\bot}$ denotes the unit
vector spanned by $\ket{\psi_0}$ and $\ket{\psi_1}$, orthogonal to $\ket{\psi_\zeta}$.
According to Lemma~\ref{lemma_nonlocal}, there exists $\kappa$ positive such that for each $x$ and $y$,
$|\gamma_{0,x}^0|^2\leq \frac{\kappa}{c}$ and $|\gamma_{0,y}^1|^2\leq \frac{\kappa}{c}$.
We define operators $\tilde{G}_{\alpha,\beta}^0$ and $\tilde{G}_{\alpha,\beta}^1$ by
\begin{eqnarray*}
\tilde{G}_{\alpha,\beta}^0 &\assign& \left(1-\frac{\kappa}{c}\right)\cdot\sqrt{\sum_{x,y}\mu_{x,y}}\proj{\psi_{1-\alpha}^\bot}+\nu_0(c)\proj{\psi_{1-\alpha}},\\
\tilde{G}_{\alpha,\beta}^1  &\assign& \left(1-\frac{\kappa}{c}\right)\cdot\sqrt{\sum_{x,y}\mu_{x,y}}\proj{\psi_{1-\beta}^\bot}+\nu_1(c)\proj{\psi_{1-\beta}}
\end{eqnarray*}
for non-negative functions $\nu_0,\nu_1\in O(1/c)$ chosen to be such that
$$p_{E^*_{\alpha,\beta}}=p_{\tilde{G}_{\alpha,\beta}^0\otimes \tilde{G}_{\alpha,\beta}^1}.$$
Such a choice of parameters is possible, due to the fact the the probability of
a wrong guess, conditioned on the outcome $E^*_{\alpha,\beta}$ is in $O(1/c)$.
Since operators $\{E^*_{\alpha,\beta}\}_{\alpha,\beta}$ form a valid POVM,
after projecting them by a projector $P\assign \proj{\psi_{1-\alpha}^\bot}\otimes \proj{\psi_{1-\beta}^\bot}$,  we get a valid POVM on the support of $P$. In other words, $\{PE^*_{\alpha,\beta}P\}_{\alpha,\beta}$ form a POVM and therefore, also the operators
$$J_{\alpha,\beta}\assign \left(1-\frac{\kappa}{c}\right)\cdot \left(\sum_{x\rightarrow \alpha,y\rightarrow \beta}\mu_{x,y}\right) \proj{\psi_{1-\alpha}^\bot}\otimes \proj{\psi_{1-\beta}^\bot},$$
lower-bounding $PE^*_{\alpha,\beta}P$,
form valid POVMs. From the condition
$$\|\sum_{\alpha,\beta}J_{\alpha,\beta}\|_\infty\leq 1,$$
we conclude that
\begin{equation}
\|\sum_{\alpha,\beta}\tilde{G}_{\alpha,\beta}^0\otimes \tilde{G}_{\alpha,\beta}^1\|_\infty\in 1+O(1/c).
\end{equation}

It remains to show that
$$
\forall \zeta_0,\zeta_1,\alpha,\beta \in\{0,1\}:\ \left|\bra{\psi_{\zeta_0},\psi_{\zeta_1}}\tilde{G}_{\alpha,\beta}^0\otimes 
\tilde{G}_{\alpha,\beta}^1\ket{\psi_{\zeta_0},\psi_{\zeta_1}}-\bra{\psi_{\zeta_0},\psi_{\zeta_1}}E^*_{\alpha,\beta}\ket{\psi_{\zeta_0},\psi_{\zeta_1}}\right|\in O(1/\sqrt{c}).$$

By definition of $\tilde{G}_{\alpha,\beta}^0$ and $\tilde{G}_{\alpha,\beta}^1$, this is true if $\zeta_0\neq \alpha$ or $\zeta_1\neq \beta$. We now discuss the remaining case. It follows from Lemma~\ref{lemma_nonlocal}, applied to each $F_x^0\otimes F_y^1$ and the construction of $\tilde{G}_{\alpha,\beta}^0\otimes \tilde{G}_{\alpha,\beta}^1$ that
$$\left\|\sum_{x\rightarrow \alpha,y\rightarrow \beta}\mu_{x,y}F_x^0\otimes F_y^1-\tilde{G}_{\alpha,\beta}^0\otimes \tilde{G}_{\alpha,\beta}^1\right\|_\infty \in O(1/\sqrt{c}).$$
Hence, also
$$\left|\bra{\psi_{\alpha},\psi_{\beta}}\tilde{G}_{\alpha,\beta}^0\otimes 
\tilde{G}_{\alpha,\beta}^1\ket{\psi_\alpha,\psi_\beta}-\bra{\psi_\alpha,\psi_\beta}E^*_{\alpha,\beta}\ket{\psi_\alpha,\psi_\beta}\right|\in O(1/\sqrt{c}).$$

We have defined a set of operators
$\{\tilde{G}_{\alpha,\beta}^0\otimes \tilde{G}_{\alpha,\beta}^1\}_{\alpha,\beta}$, almost forming a POVM
 due to the condition (iii). Therefore, we can re-scale the elements of the set
by a factor in $1-O(1/c)$, and thereby create a POVM
$\{G_{\alpha,\beta}^0\otimes G_{\alpha,\beta}^1\}_{\alpha,\beta}$. Due to the
condition (i), the expected payoffs conditioned on measuring either
$E^*_{\alpha,\beta}$ or $G_{\alpha,\beta}^0\otimes G_{\alpha,\beta}^1$ are the same.
Finally, due to the condition (ii), the probabilities of measuring an
outcome from $\{E^*_{\alpha,\beta}\}_{\alpha,\beta}$ and an outcome from
$\{G_{\alpha,\beta}^0\otimes G_{\alpha,\beta}^1\}_{\alpha,\beta}$ differ by a value in $O(1/\sqrt{c})$. Hence,
if the probability of a conclusive answer of $\mathcal{M}^*$ is constant
then the measurement with elements
$$
E''_0\assign G_{0,0}^0\otimes G_{0,0}^1+G_{1,1}^0\otimes G_{1,1}^1,\ \ \ \ 
E''_1\assign G_{0,1}^0\otimes G_{0,1}^1+G_{1,0}^0\otimes G_{1,0}^1,\ \ \ \ 
E''_?\assign \id -E'_0-E'_1
$$
gives a conclusive answer with probability lower by at most a value in $O(1/\sqrt{c})$, and differs from $\mathcal{M}^*$ in its payoff by a value in $O(1/\sqrt{c})$.
According to~\cite{KKB05}, the state of each of the two subsystems after applying the measurement given above
is independent of the outcome in the other one. Therefore, in order to achieve certain expected
payoff, the local measurements can be optimized separately. It follows that the payoff
of measurement $(E''_0, E''_1, E''_?)$ can be matched by the payoff $p'$
of some measurement $\mathcal{M}'$ in the form:
$$E'_0=G^0_0\otimes G^1_0+G^0_1\otimes G^1_1, \ \ \ E'_1=G^0_0\otimes G^1_1+G^0_1\otimes G^1_0, \ \ \ E'_?=\id-E'_0-E'_1.$$

By applying~(\ref{differ1k}) and the triangle inequality, we finally get
that
$$|p-p'|\in O(1/k)+O(1/\sqrt{c}).$$
\end{proof}

\end{document}